\RequirePackage{fix-cm}
\documentclass[smallextended]{svjour3}       
\journalname{}
\smartqed  
\usepackage{graphicx}

\usepackage{amssymb,amsmath}
\usepackage{enumerate}
\usepackage{lineno}
\usepackage{url}

\usepackage{subfigure}

\def\R{{\mathbb R}}

\def\bd{{\partial}}
\def\dd{{d}}

\def\D{{\cal D}}
\def\E{{\cal E}}

\def\DD{\D_{B,A}}
\DeclareMathOperator{\VD}{{\cal V}}
\DeclareMathOperator{\h}{{H}}
\newcommand{\scalprod}[2]{\left\langle #1,#2 \right\rangle}

\newcommand{\N}{\mathbb N}

\usepackage{color}

\usepackage{hyperref} 

\usepackage[textsize=tiny]{todonotes}

\begin{document}
\title{Partial-Matching and Hausdorff RMS Distance Under Translation: Combinatorics and Algorithms}
\titlerunning{Partial-Matching and Hausdorff RMS Distance}  

\author{Rinat Ben-Avraham\and Matthias Henze \and Rafel Jaume \and Bal\'azs Keszegh \and Orit~E.~Raz\and Micha Sharir \and Igor Tubis}

\authorrunning{Rinat Ben-Avraham et al.} 
\institute{Rinat Ben-Avraham, Orit E.~Raz, Micha Sharir, Igor Tubis \at
              Blavatnik School of Computer Science, Tel Aviv University, Tel~Aviv, Israel \\              
              \email{rinatba@gmail.com,\{oritraz,michas\}@post.tau.ac.il, mrtubis@gmail.com}
							\and
						Matthias Henze, Rafel Jaume\at
							Institut f\"ur Informatik, Freie Universit\"at Berlin, Berlin, Germany\\
							\email{matthias.henze@fu-berlin.de, jaume@mi.fu-berlin.de}
							\and
						Bal\'azs Keszegh\at
							Alfr\'ed R\'enyi Institute of Mathematics, Hungarian Academy of Sciences, Budapest, Hungary\\
							\email{keszegh@renyi.hu}
}

\date{\today}

\maketitle              

\begin{abstract}
We consider the RMS distance (sum of squared distances between pairs of points) under translation between two point sets in the plane, in two different setups.
In the partial-matching setup, each point in the smaller set is matched to a distinct point in the bigger set. 
Although the problem is not known to be polynomial, we establish several structural properties of the underlying subdivision of the plane and derive improved bounds on its complexity.
These results lead to the best known algorithm for finding a translation for which the partial-matching RMS distance between the point sets is minimized. 
In addition, we show how to compute a local minimum of the partial-matching RMS distance under translation, in polynomial time.
In the Hausdorff setup, each point is paired to its nearest neighbor in the other set. 
We develop algorithms for finding a local minimum of the Hausdorff RMS distance in nearly linear time on the line, 
and in nearly quadratic time in the plane. These improve substantially 
the worst-case behavior of the popular ICP heuristics for solving this problem.

\keywords{Shape matching, partial matching, Hausdorff RMS distance, convex subdivision, local minimum}
\end{abstract}

\section{Introduction}

Let $A$ and $B$ be two finite sets of points in the plane, of respective cardinalities $n$ and $m$.
We are interested in measuring the similarity between $A$ and~$B$, under a suitable proximity measure. 
We consider two such measures where the proximity is the sum of the squared distances between pairs of points. 
We refer to the measured distance between the sets, in both versions, as the \emph{RMS (Root-Mean-Square) distance}.
In the first, we assume that $n \ge m$ and we want to match all the points of $B$ (a specific pattern that we want to identify), in a one-to-one manner, to a subset of $A$ (a larger picture that ``hides'' the pattern) of size $|B|$. 
This measure is called the \emph{partial-matching RMS distance}. 
This is motivated by situations where we want a one-to-one matching between features of $A$ and features of $B$~\cite{JL,leastSquares,ZikanSilberberg}.
In the second, each point is assigned to its nearest neighbor in the other set, and we call
it the \emph{Hausdorff RMS distance}. 
See~\cite{balls} for a similar generalization of the Hausdorff distance.
In both variants the sets $A$ and $B$ are in general not aligned, so we seek a translation of 
one of them that will minimize the appropriate RMS distance, partial matching or Hausdorff. 

\subsection{\bf The partial-matching RMS distance problem}
Let $A=\{a_1,\ldots,a_n\}$ and $B=\{b_1,\ldots,b_m\}$ be two sets of points in the plane.
Here we assume that $m \le n$, and we seek a minimum-weight \emph{maximum-cardinality matching} of 
$B$ into $A$. 
This is a subset $M$ of edges of the complete bipartite graph with edge set
$B\times A$, so that each $b\in B$ appears in exactly one edge of $M$, and each $a\in A$ 
appears in at most one edge. 
The weight of an edge $(b,a)$ is $\|b-a\|^2$, the squared Euclidean distance between $b$ and $a$, and the weight 
of a matching is the sum of the weights of its edges.

A maximum-cardinality matching can be identified as an injective assignment $\pi$ of $B$ into $A$. 
With a slight abuse of notation, we denote by $a_{\pi(i)}$ the point 
$a_j$ that $\pi$ assigns to $b_i$. In this notation, the minimum RMS partial-matching 
problem (for fixed locations of the sets) is to compute
\[
M(B,A) = \min_{\pi:B\rightarrow A\textrm{ injective}} 
\sum_{i=1}^m \left\| b_i - a_{\pi(i)}\right\|^{2}.
\]
Allowing the pattern $B$ to be translated, we obtain the problem of computing the minimum partial-matching RMS distance under translation, 
defined as 
\[
M_T(B,A) = \min_{t\in \R^2}M(B+t,A) =\min_{\substack{t\in \R^2,\pi:B\rightarrow A,\\ \pi\textrm{ injective}}} 
\sum_{i=1}^m \left\|b_i + t- a_{\pi(i)}\right\|^{2}.
\] 
Here $a_{\pi(i)}$ is the point of $A$ assigned to $b_i+t$; in this notation we hide the explicit dependence on $t$, and assume it to be understood from the context.

The function $F(t):=M(B+t,A)$ 
induces a subdivision of $\R^2$, where two points $t_1,t_2\in\R^2$ are in the same region if the minimum of $F$ at $t_1$ and at $t_2$ are attained by the same set of
assignments. 
We refer to this subdivision, following Rote~\cite{roteNote1}, 
as the \emph{partial-matching subdivision} and denote it by $\D_{B,A}$. We say that a matching is \emph{optimal} if it attains $F(t)$ for some $t \in \R^2$.

\subsection{\bf  The Hausdorff RMS distance problem}
Let $N_A(x)$ (resp., $N_B(x)$) denote a nearest neighbor in $A$ (resp., in $B$) of a point 
$x\in{\R^2}$ breaking ties arbitrarily. The \emph{unidirectional (Hausdorff) RMS distance} between $B$ and $A$ is defined as
\begin{linenomath}
\[
\h(B,A) = \sum_{b\in B} \|b-N_A(b)\|^{2}.
\]
\end{linenomath}
We also consider \emph{bidirectional} RMS distances, in which we also measure distances 
from the points of $A$ to their nearest neighbors in $B$. We consider two variants of this 
notion. The first variant is the \emph{$L_1$-bidirectional RMS distance} between $A$ and $B$, 
which is defined as 
\begin{linenomath}
\[
\h_1(B,A) = \h(A,B) + \h(B,A).
\]
\end{linenomath}
The second variant is the \emph{$L_\infty$-bidirectional RMS distance} between $A$ and $B$, and is defined as
\begin{linenomath}
\[
\h_\infty(B,A) = \max{\left\{ \h(A,B), \h(B,A) \right\}}.
\]
\end{linenomath}
Allowing one of the sets (say, $B$) to be translated, we define the \emph{minimum unidirectional 
RMS distance under translation} to be
\begin{linenomath}
\[
\h_T(B,A) = \min_{t\in \R^2} \h(B+t,A) =\min_{t\in \R^2} \sum_{b\in B} \|b + t - N_A(b + t)\|^{2} ,
\]
\end{linenomath}
where $B+t=\{b_1+t,\ldots,b_m+t\}$.
Similarly, we define the \emph{minimum $L_1$- and $L_\infty$-bidirectional RMS distances 
under translation} to be 
\begin{linenomath}
\begin{align*}
\h_{T,1}(B,A) &= \min_{t\in \R^2}\h_1(B+t,A)\qquad\textrm{and}\\
\h_{T,\infty}(B,A) &= \min_{t\in \R^2}\h_\infty(B+t,A). 
\end{align*}
\end{linenomath}

\subsection{\bf Background}
A thorough initial study of the minimum RMS partial-matching distance under translation is given by Rote~\cite{roteNote1};
see also \cite{rotePaper,roteNote2} for two follow-up studies, another study in \cite{AgarwalPhillips},
and an abstract of an earlier version of parts of this paper~\cite{HJK}. 
The resulting subdivision $\D_{B,A}$, as defined above,
is shown in \cite{roteNote1} to be a subdivision whose faces are convex polygons. 
Rote's main contribution for the analysis of the complexity of $\D_{B,A}$ was to show
that a line crosses only $O(nm)$ regions of the subdivision (see Theorem~\ref{thm:Rote} below). 
However, obtaining sharp bounds for the complexity of $\D_{B,A}$, in particular, settling whether this complexity is polynomial or not, is 
still an open issue.	

The problem of Hausdorff RMS minimization under translation has been considered in 
the literature (see, e.g., \cite{balls} and the references therein), although only scarcely so (compared to the extensive body of work on the standard \mbox{$L_\infty$-Hausdorff} distance). 
If $A$ and $B$ are sets of points on the line, the complexity of the Hausdorff RMS function, 
as a function of the translation $t$, is $O(nm)$ (and this bound is tight in the worst case). Moreover,
the function can have many local minima (up to $\Theta(nm)$ in the worst case). Hence,
finding a translation that minimizes the Hausdorff RMS distance can be done in 
brute force, in $O(nm\log (nm))$ time, but a worst-case near linear algorithm is not known. 
In practice, though, there exists a popular heuristic technique, called the ICP 
(Iterated Closest Pairs) algorithm, proposed by Besl and McKay~\cite{ICP92}
and analyzed in Ezra et al.~\cite{ICP2006}.
Although the algorithm is reported to be efficient in practice, it might perform $\Theta(nm)$ 
iterations in the worst case. Moreover, each iteration takes close to linear time (to find 
the nearest neighbors in the present location).

The situation is even worse in the plane, where the complexity of the Hausdorff 
RMS function is $O(n^2m^2)$, a bound which is worst-case tight, and the bounds
for the performance of the ICP algorithm, are similarly worse.
Similar degradation shows up in higher dimensions too; see, e.g., \cite{ICP2006}.

\subsection{\bf  Our results}
In this paper we study these two fairly different variants of the problem of minimizing
the RMS distance under translation, and improve the state of the art in both of them.

In the partial-matching variant, we first analyze
the complexity of the subdivision $\D_{B,A}$. 
We significantly improve the bound from the naive $O(n^m)$ to
$O(n^2m^{3.5} (e \ln m+e)^m)$. 
In particular, the complexity is only quadratic in the size of the larger set, albeit still slightly superexponential in the size of the smaller set. 
In addition, we provide the best known lower bound $\Omega(m^2(n-m)^2 )$ for the complexity of~$\D_{B,A}$. 
Albeit being only polynomial in $m$, this bound is tight with respect to $n$.

A preliminary informal exposition of this analysis by a subset of the authors
is given in the (non-archival) note \cite{HJK}.
The present paper expands the previous note,
derives additional interesting structural properties of the subdivision, and significantly
improves the complexity bound. 
The arguments that establish the bound can be generalized to bound the number of regions (full-dimensional cells) of the 
analogous subdivision in $\R^d$ by\footnote{Note that in two dimensions the first factor improves to $(nm^2)^d=(nm^2)^2$. See below for details.} $O\left( (n^{2}m)^d (e \ln m+e)^m)/ \sqrt{m} \right)$. 
The derivation of the upper bound proceeds by a reduction that connects partial matchings 
to a combinatorial question based on a game-theoretical problem, which we believe to be
of independent interest.

Next we present a polynomial-time algorithm, that runs in $O(n^3m^6\log{n})$ time, for finding a local minimum of the
partial-matching RMS distance under translation. 
This is significant, given that we do not have a polynomial
bound on the size of the subdivision. 
We also fill in the details of explicitly computing the intersections of a line with the edges and faces of $\DD$. 
Rote hinted at such an algorithm in \cite{roteNote1} but, exploiting some new properties of $\DD$ derived here, we manage to compute the intersections in a simple, more efficient manner.

We also note that by combining the combinatorial bound for the complexity of $\DD$, along with the procedures in the algorithm for finding a local minimum of the partial-matching RMS distance, it is possible to traverse all of $\DD$, and compute a \emph{global} minimum of the partial-matching RMS distance in time $O(n^3 m^{7.5} (e \ln{m}+e)^m)$. This is the best known bound for this problem. 

For the Hausdorff variant, we provide improved algorithms for computing a local minimum
of the RMS function, in one and two dimensions. Assuming $|A|=|B|=n$, in the one-dimensional 
case the algorithms run in time $O(n\log^2n)$, and in the two-dimensional case they run in time $O(n^2\log^2 n)$. 
Our approach thus beats the worst-case running time of the ICP algorithm (used for about two decades to solve this problem). 
The approach is an efficient search through the (large number of) critical values of the RMS function. 
The techniques are reasonably standard, although their assembly is somewhat involved. 
This part of the work was partially presented in~\cite{EinGedi}.

Note that in both the partial matching and the Hausdorff variants, our algorithms 
primarily compute a {\em local} minimum of the corresponding objective function. It 
would naturally be more desirable to find the
{\em global} minimum, but we do not know how to achieve this without an exhaustive 
search through all possible combinatorially different translations (in both cases), which would make the 
algorithms considerably less efficient, especially in the partial-matching setting, as noted 
above. Concerning the effectiveness of computing a local minimum, we make the 
following comments.

In the Hausdorff variant, the ICP algorithm, which we want to improve, also constructs 
only a local minimum. In this regard, the developers of ICP and other users of the technique 
have proposed several heuristic enhancements (for a collection of these see \cite{ICP92,fasticp}). 
In particular, starting the algorithm at a translation that is sufficiently close to the global 
minimum is likely to converge to that minimum. In real life applications, usually this can be done by human assistance. Avoiding such manual input, an alternative exploitation of this idea is to start the algorithm from several scattered initial placements, and hope that one of these incarnations will lead to the global minimum.

These, and other similar heuristics can be applied, with suitable modifications, to our algorithms too. 
See a remark to that effect, later on in Section~\ref{subsec:pmalg}. 

In particular, if the points of $A$ and of $B$ are in general, sufficiently separated positions, one would expect the local minima to be sparse and sufficiently separated, so the strategy of using a reasonable number of starting ``seeds'' that are well separated and cover the relevant portion of the translation space, has good chances of hitting the global minimum. Alternatively, it could be the case that the pattern $B$ appears several times in $A$, so that each such appearance has its own local minimum. In such cases it usually does not matter which local minimum we hit, assuming that they are all more or less of the same matching quality.

Other heuristics have also been proposed for the ICP algorithm. For example,
one could sample smaller subsets of the original point sets, run an inefficient algorithm for finding the global minimum for the samples, and use that as a starting translation. Heuristics of this kind can also be applied to our algorithms.

\section{Properties of \texorpdfstring{$\D_{B,A}$}{D\_\{B,A\}}}
\label{properties}

We begin by reconstructing several basic properties of $\D_{B,A}$ that have been noted by Rote in~\cite{roteNote1}.
First, if we fix the translation $t \in{\R^2}$ and the assignment~$\pi$, the cost of the matching, denoted by $f(\pi,t)$, is
\begin{linenomath}
\begin{align}
 f(\pi,t) &=  \sum_{i=1}^m \left\| b_i + t - a_{\pi(i)}\right\|^{2} =  c_\pi + \scalprod {t}{\dd_\pi } + m \left\| t\right\|^{2} \label{eq:Rote},
\end{align}
\end{linenomath}
where $c_\pi=\sum_{i=1}^m \left\| b_i - a_{\pi(i)}\right\|^{2}$ and $\dd_\pi=2\sum_{i=1}^m  (b_i - a_{\pi(i)})$. 
For $t$ fixed, the assignments that minimize $f(\pi,t)$ are the same assignments 
that minimize $(\pi,t)\mapsto c_\pi+ \scalprod{ t}{ \dd_\pi }$. 
It follows 
that $\D_{B,A}$ is the minimization diagram (the $xy$-projection) of the graph of the function
\[
\E_{B,A}(t) = \min_{\pi:B\to A\textrm{ injective}} \left (c_\pi + \scalprod{t}{\dd_\pi} \right),\quad t\in\R^2.
\]
This is the lower envelope of a finite number of planes, so its graph is a convex polyhedron, 
and its projection $\D_{B,A}$ is a convex subdivision of the plane, whose faces are convex polygons.
In particular, it follows that an assignment $\pi$ can be associated with at most one open region 
of the subdivision $\D_{B,A}$.

The great open question regarding minimum partial-matching RMS distance under translation, 
is whether the number of regions of $\D_{B,A}$ is polynomial in~$m$ and~$n$. 
A significant, albeit small step towards settling this question is the 
following result of Rote~\cite{roteNote1}.
\begin{theorem}[Rote \cite{roteNote1}]\label{thm:Rote}
A line intersects the interior of at most $m(n-m)+1$ different regions of the 
partial-matching subdivision $\D_{B,A}$.
\end{theorem}  

\begin{figure}
\centering
\includegraphics{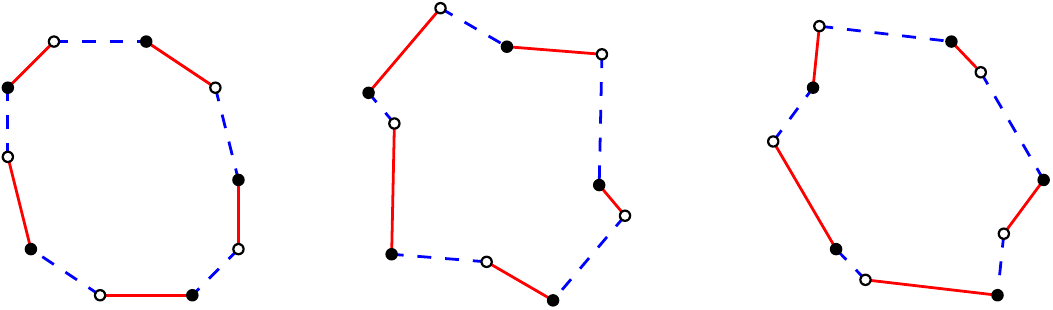}
\caption{Two optimal matchings, represented for three translations of $B$.}
\label{fig:lexbotcycle}
\end{figure}	

Note that, even for $t$ interior to a two-dimensional face of $\D_{B,A}$, more than one matching can attain $\E_{B,A}(t)$.
That is, two different matchings can have equal cost along an open set (and hence, everywhere) and be optimal in it. 
Indeed, the vector $\dd_\pi$ depends only on the centroid of the matched set.  
Therefore, if two matchings have matched sets with the same centroid, and they have the same cost for some translation $t_0\in \R^2$, the matchings have the same cost everywhere.
The solid and the dashed matchings displayed in Figure~\ref{fig:lexbotcycle} are thus both optimal over an open neighborhood of the depicted translations. 
Observe that these matchings use the same subset of points in~$A$. 
It will be shown later that this is in fact a necessary condition for two matchings to be simultaneously optimal in the same open set.

The following property, observed by Rote~\cite{roteNote1}, seems to be well known~\cite{ZikanSilberberg}.

\begin{lemma}\label{lem4_sameSize1Match}
For any $A' \subseteq A$, with $|A'|=m$, the optimal assignments that realize the minimum $M(B+t,A')$ are independent of the translation $t \in{\R^2}$.
\end{lemma}
\begin{proof}
For $A'$ fixed, $d_\pi$ is independent of $\pi$, so any assignment $\pi$ that minimizes $c_\pi=\E_{B,A}(0)$ minimizes $\E_{B,A}(t)$ for any $t$.  \qed
\end{proof}

Next, we derive several additional properties of $\D_{B,A}$ which show that the diagram has, locally, low-order polynomial complexity. 

\begin{lemma}\label{lemdirections}
Every edge of $\D_{B,A}$ has a normal vector of the form $a_j-a_i$ for suitable $i\neq j\in\{1,\ldots,n\}$.
\end{lemma}

\begin{proof}
Let $E$ be an edge common to the regions associated with the injections $\pi,\sigma:B\to A$.
By definition, $f(\pi,t)=f(\sigma,t)\leq f(\delta,t)$ for every injection $\delta:B\to A$ and for any $t\in E$. 
So $E$ is contained in the line
\[\ell(\pi,\sigma):=\{t\in\R^2\mid\scalprod{t}{d_\pi - d_\sigma}=c_\sigma - c_\pi\}.\]
Let $\pi \triangle \sigma=(\pi \setminus \sigma) \cup (\sigma \setminus \pi)$.
It is easy to see that $\pi \triangle \sigma$ 
consists of a vertex-disjoint union of alternating cycles and alternating paths; that is, the edges of each cycle and path alternate between edges of $\pi$ and edges of $\sigma$. 
Also, each path (and, trivially, each cycle) is of even length.
Let $\gamma_1,\ldots,\gamma_p$ be these 
cycles and paths. 
Every cycle and every path can be ``flipped'' 
independently while preserving the validity of the matching; that is, we can choose, within
any of the $\gamma_j$'s, either all the edges corresponding to $\pi$ or all the ones 
corresponding to $\sigma$, without interfering with other cycles or paths, so that the
resulting collection of edges still represents an injection from $B$ into $A$. 
Observe now that $\ell(\pi,\sigma)=\{t\in\R^2\mid \langle t, \sum_{j=1}^p  d_{\gamma_j} \rangle
 = - \sum_{j=1}^p c_{\gamma_j} \}$, where $d_{\gamma_j}$ is the sum of the terms in 
$d_\pi- d_\sigma$ that involve only the $a_i \in A$ contained in $\gamma_j$ and 
$c_{\gamma_j}$ is analogously defined for $c_\pi - c_\sigma$. Note that $d_{\gamma_j}$ 
is $0$ for every cycle $\gamma_j$ and, therefore, at least one of the $\gamma_j$'s is a path. 
(Otherwise, $d_\pi-d_\sigma=0$ and $\ell(\pi,\sigma)$ is either empty or the entire plane, contrary to our assumptions.)
Then, we must have $\scalprod{t}{d_{\gamma_j} }= -c_{\gamma_j}$ for all $j=1,\ldots,p$ 
and every $t\in \ell(\pi,\sigma)$. Otherwise, a flip in a path or a cycle violating this	 equation 
would contradict the optimality of $\pi$ or of $\sigma$ along $\ell(\pi,\sigma)$. Therefore, 
all the nonzero vectors $d_{\gamma_j}$ must be orthogonal to $\ell(\pi,\sigma)$. Hence, the direction of 
$d_\pi - d_\sigma$ is the same as the one of $d_{\gamma_j}$ for every \emph{path}~$\gamma_j$. 
If a path, say $\gamma_1$, starts at some $a_j$ and ends at some $a_i$, then 
$d_{\gamma_1} = a_j-a_i$, which concludes the proof. \qed
\end{proof}

\smallskip
\noindent{\textit{Remarks}. }
It follows from the proof of Lemma~\ref{lemdirections} that $d_\gamma$ is $0$ if and only if $\gamma$ is a cycle. 
This fact implies that, if two matchings have the same cost over an open set, their symmetric difference consists only of cycles, 
which implies that they match the same subset of $A$, as claimed above. 
Note also that if $A$ is in general position then $\pi \triangle \sigma$ 
has exactly one alternating path, and the pair $a_i$, $a_j$ is unique.

\begin{lemma}\label{lem:structure}
\begin{enumerate}[(i)]
 \item\label{lem1_UnboundedCells} $\D_{B,A}$ has at most $4m(n-m)$ unbounded regions.
 \item\label{lem_bound_edges} Every region in $\D_{B,A}$ has at most $m(n-m)$ edges.
 \item\label{cor1_degree} Every vertex in $\D_{B,A}$ has degree at most $2m(n-m)$.
 \item\label{lem6_convexPath} Any convex path can intersect at most $m(n-m) + n(n-1)$ regions of $\D_{B,A}$.
\end{enumerate}
\end{lemma}
\begin{proof}
(\romannumeral1) Take a bounding box that encloses all the vertices of the diagram. By Theorem~\ref{thm:Rote}, 
every edge of the bounding box crosses at most $m(n-m)+1$ regions of $\D_{B,A}$. The edges of the 
box traverse only unbounded regions, and cross every unbounded region exactly once, except for the
coincidences of the last region traversed by an edge and the first region traversed by the next edge.

(\romannumeral2) 
By Lemma~\ref{lemdirections}, the normal vector of every edge of a region corresponding to an injection $\pi$ is a multiple of 
$a_j-a_i$ for some $a_i\in\pi(B)$ and $a_j\notin\pi(B)$. 
There are exactly $m(n-m)$ such possibilities.

(\romannumeral3) Let $v$ be a vertex of $\D_{B,A}$. Draw two generic parallel lines close enough to each other to enclose~$v$ and no other vertex.
Each edge adjacent to $v$ is crossed by one of the lines, and by Theorem~\ref{thm:Rote} each of these lines crosses at most $m(n-m)$ edges.

(\romannumeral4) We use the following property that was observed in Rote's proof of Theorem~\ref{thm:Rote}.
Suppose that we translate $B$ along a line in some direction $v$.
Rank the points of $A$ by their order in the $v$-direction, i.e., $a < a'$ means that $\scalprod{a}{v} < \scalprod{a'}{v}$
(for simplicity, assume that $v$ is generic so there are no ties).
Let $\Phi$ denote the sum of the ranks of the $m$ points of $A$ that participate in an optimal partial matching. 
As Rote has shown, whenever the set matched by the optimal assignments changes, $\Phi$ must increase.

Now follow our convex path $\gamma$, which, without loss of generality, can be assumed to be
polygonal. As we traverse an edge of $\gamma$, $\Phi$ obeys the above property, increasing
every time we cross into a new region of $\D_{B,A}$.
 When we turn (counterclockwise) at a 
vertex of $\gamma$, the ranking of $A$ may change, but each such change 
consists of a sequence of swaps of consecutive elements in the present ranking. 
At each such swap, $\Phi$ can decrease by 
at most~$1$. Since $\gamma$ is convex, each pair of points of $A$ can be 
swapped at most twice, so the total decrease in $\Phi$ is at most $2\binom{n}{2}=n(n-1)$.

In conclusion, the accumulated increase in $\Phi$, and thus also the total number of regions of $\D_{B,A}$ 
crossed by $\gamma$, is at most
\[
\Bigl( n + (n-1) + \ldots + (n-m+1) \Bigr) - \Bigl( 1 + 2 + \ldots + m \Bigr) + n(n-1)= m(n-m)+n(n-1),
\]
as desired.\qed
\end{proof}

\noindent{\em Remark.} 
In order to gain better understanding of how the potential $\Phi$, introduced in the proof of 
Lemma~\ref{lem:structure}(\ref{lem6_convexPath}), changes when the set $B$ is translated 
along a convex path, consider a standard dual construction~\cite{deBerg}, where the points $a\in{A}$ 
are mapped to lines $a^*$ in the following manner: 
$$
a=(a_x,a_y) \longmapsto y= a_y x + a_x .
$$
The duality is order preserving in the sense that, given two points $a_1, a_2$, and a direction 
$u=(u_x,u_y)$ with $u_x>0$ in the primal plane, then it can be easily checked that $\scalprod{a_2}{ u} > \scalprod{a_1 }{ u}$ 
if and only if $a_2$ is above $a_1$ at the $x$-coordinate that corresponds to $u_y/u_x$ in the dual plane, 
i.e., to the direction of $u$. Thus, we get an arrangement of $n$ lines, in which the heights of 
the lines at any \mbox{$x$-coordinate} in the dual plane represent the order of the points of 
$A$ along the corresponding direction in the primal plane; see Figure~\ref{fig:dualLines} for an illustration.

\begin{figure}[ht]
\begin{center}
\includegraphics{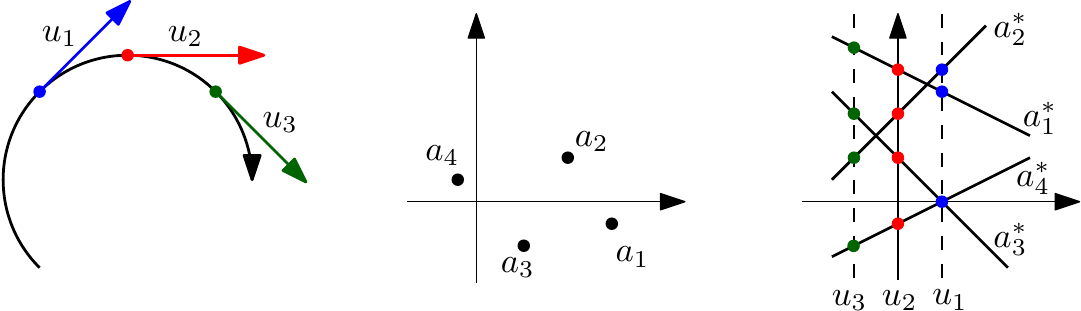}
\caption{An example for a convex path with three highlighted directions, a set of four given points, 
and the resulting dual arrangement. The order $(a_4,a_3,a_2,a_1)$ of the points in the primal 
$x$-direction corresponds to the order of the intercepts of the dual lines on the $y$-axis.}
\label{fig:dualLines}
\end{center}
\end{figure}

Furthermore, for each point on a convex path, we can mark at the corresponding $x$-coordinate in the dual plane, the $m$ dual lines that correspond to the $m$ points which are currently (optimally) matched. 
By the observations in the proofs for Rote's Theorem~\ref{thm:Rote} in~\cite{roteNote1}, if the subset $A'$ of $A$ matched by the optimal matching changes, 
it must be that some point $a_-\in A'$ is replaced by a point $a_+\in A\setminus A'$ further in the direction of the line. 
Indeed, there is such a pair for each path in the symmetric difference of the old and the new matchings. 
In our dual setting, it simply means that when a matching changes, the points sitting on the dual lines of the matched points of $A$ can only 
skip upwards to a line that passes above them. 
The sum of the indices of the marked lines (those that participate in the matching) is exactly the sum of the ranks defined in the proof of Theorem~\ref{thm:Rote}.

This dual setting also demonstrates how and when $\Phi$ could drop --- it  happens in a direction orthogonal to the direction dual to an intersection point of two dual lines, and thus the height of a matched point (its rank) can drop by 1. If one can bound the amount of such drops, i.e., for $m$ points moving along the $n$ dual lines, from left to right, skipping from line to line only upwards, then it immediately gives a bound for the number of intersections of a convex path with $\DD$. 
Unfortunately, an almost quadratic lower bound for the complexity of such monotone paths was in fact presented in~\cite{noHope}, and thus
 it seems hopeless to get a significantly better upper bound for the amount of intersections of a convex path with $\DD$ without exploiting any additional geometric properties.

\section{Bounds on the complexity of \texorpdfstring{$\D_{B,A}$}{D\_\{B,A\}}}	

In this section, we focus on establishing a global bound on the complexity of the diagram $\D_{B,A}$.
We begin by deriving the following technical auxiliary results.

\begin{lemma}\label{mainlemma}
Let $\pi$ be an optimal matching for a fixed translation $t \in \R^2$ such that 
$\Vert b_p + t - a_r \Vert \ne \Vert b_p + t - a_s \Vert$  for all $p \in \{1,\ldots,m\}$ and $r,s \in \{1,\ldots,n \}$ with $r\neq s$.
\begin{enumerate}[(i)]
 \item\label{nocycle} There is no cyclic sequence $(i_1,i_2,\ldots,i_{k},i_1)$ satisfying \\
 $\Vert b_{i_j} + t - a_{\pi(i_{j})} \Vert > \Vert b_{i_j} + t - a_{\pi(i_{j+1})} \Vert$ for all $j \in \{1,\ldots,k\}$ (modulo $k$).
 \item\label{lem2_OnlyMRelevant} Each point of $B+t$ is matched to one of its $m$ nearest neighbors in $A$.
 \item\label{lem3_1ToNN} At least one point in $B+t$ is matched to its nearest neighbor in $A$.
  \item\label{cor2_ordereing} There exists an ordering $(b_1,\ldots,b_m)$ of the elements of $B$, such that each $b_k+t$ is assigned by $\pi$ to its nearest neighbor in $A\setminus\{(b_{\pi(1)},\ldots,b_{\pi(k-1)}\}$, for $k = 1,\ldots,m$.  In particular, $b_k$ is assigned to one of its $k$ nearest neighbors in $A$, for $k=1,\ldots, m$.
\end{enumerate}
\end{lemma}
\begin{proof}
(\romannumeral1) For the sake of contradiction, we assume that there exists a cyclic sequence that satisfies all the prescribed inequalities.
Consider the assignment~$\sigma$ defined by $\sigma(i_j)=\pi(i_{j+1})$ for all $j \in \{1,\ldots,k\}$ (modulo $k$) and $\sigma(\ell)=\pi(\ell)$ for all other indices $\ell$.
Since $\pi$ is a one-to-one matching, we have that $\pi(i_{j}) \neq \pi(i_{j'})$ for all distinct $j,j' \in \{1,\ldots,k\}$ and, consequently, $\sigma$ is
one-to-one as well. It is easily checked that $f(\sigma,t)<f(\pi,t)$, contradicting the optimality of~$\pi$.

(\romannumeral2) For contradiction, assume that, for some point $b \in{B}$, $b+t$ is not matched by $\pi$ to one of its $m$ nearest neighbors in~$A$.
Then, at least one of these neighbors, say $a$, cannot be matched (because these $m$ points can be claimed only by the remaining $m-1$ points of $B+t$).
Thus, we can reduce the cost of $\pi$ by matching $b+t$ to $a$, a contradiction that establishes the claim.

(\romannumeral3) Again we assume for contradiction that $\pi$ does not match any of the points of $B+t$ to its nearest neighbor in~$A$.
We construct the following cyclic sequence in the matching~$\pi$.
We start at some arbitrary point $b_1 \in{B}$, and denote by $a_1$ its nearest neighbor in $A$ (to simplify the presentation, we do not explicitly mention the translation $t$ in what follows).
By assumption, $b_1$ is not matched to $a_1$.
If $a_1$ is also not claimed in $\pi$ by any of the points of $B$, then $b_1$ could have claimed it, thereby reducing the cost of $\pi$, which is impossible.
Let then $b_2$ denote the point that claims $a_1$ in $\pi$.
Again, by assumption, $a_1$ is not the nearest neighbor $a_2$ of $b_2$, and the preceding argument then implies that $a_2$ must be claimed by some other point $b_3$ of $B$.
We continue this process, and obtain an alternating path $\left(b_1,a_1,b_2,a_2,b_3,\ldots\right)$ such that 
the edges $(b_i,a_i)$ are not in $\pi$, and the edges $(b_{i+1},a_i)$ belong to $\pi$, for $i=1,2,\ldots$.
The process must terminate when we reach a point $b_k$ that either coincides with $b_1$, or is such that its nearest neighbor is among the already encountered points $a_i$, $i<k$.
We thus obtain a cyclic sequence as in part~(\ref{nocycle}), reaching a contradiction.

(\romannumeral4) Start with some point $b_1\in{B}$ such that $b_1+t$ goes to its nearest neighbor $a_1$ in $A$ in the optimal partial matching $\pi$; such a point exists by part~(\ref{lem3_1ToNN}).
Delete $b_1$ from $B$, and $a_1$ from $A$.
The restriction of $\pi$ to the points in $B\setminus \left\{{b_1}\right\}$ is an optimal matching for $B\setminus \left\{{b_1}\right\}$ and $A\setminus \left\{{a_1}\right\}$ (relative to~$t$), because otherwise we could have improved $\pi$ itself.
We apply part~(\ref{lem3_1ToNN}) to the reduced sets, and obtain a second point $b_2\in{B\setminus \left\{{b_1}\right\}}$ whose translation $b_2+t$ is matched to its nearest neighbor $a_2$ in $A\setminus \left\{{a_1}\right\}$, which is either its first or second nearest neighbor in the original set~$A$.
We keep iterating this process until the entire set $B$ is exhausted.
At the $k$-th step we obtain a point $b_k\in{B\setminus\left\{{b_1,\ldots,b_{k-1}}\right\}}$, such that the nearest neighbor $a_k$ in $A\setminus \left\{{a_1,\ldots,a_{k-1}}\right\}$ is matched to $b_k$ by $\pi$. \qed
\end{proof}

Observe that the geometric properties in Lemma~\ref{mainlemma} can be interpreted in purely combinatorial terms.
That is, for $t$ fixed, associate with each $b_i\in B$ an ordered list $L_t(b_i)$, called its \emph{preference list},
which consists of the points of $A$ sorted by their distances from $b_i+t$. 
A matching $\pi: B \hookrightarrow A$ is said to be \emph{better} than another (distinct) matching $\sigma:B \hookrightarrow A$ if 
either $\pi(b)=\sigma(b)$ or $\pi(b)$ appears before $\sigma(b)$ in the preference list of $b$, for each $b \in B$. 
A matching is called \emph{Pareto efficient} (hereafter, \emph{efficient}, for short) if there is no better matching. 
For the balanced case,\footnote{For the case $m=n$, the problem of finding an efficient matching was studied in the game theory literature under the name of the 
\emph{House Allocation Problem}.} 	
where $m=n$, being efficient is equivalent to the non-existence of a cycle as in Lemma~\ref{mainlemma}(\ref{nocycle}) (see~\cite{ShapleyScarf}).
In the unbalanced case, where we have $m<n$ ordered lists on $n$ elements, these properties are not equivalent. 
However, we now give a simple proof of the fact that optimal matchings are efficient.

\begin{lemma}\label{lemmaStabPref}
Let $t \in \R^2$ be such that $\Vert b_p + t - a_r \Vert \ne \Vert b_p + t - a_s \Vert$  for all $p \in \{1,\ldots,m\}$ and $r,s \in \{1,\ldots,n \}$ with $r\neq s$.
Every optimal matching for $t$ is an efficient matching for the corresponding set $L= \{L_t(b_i) \mid i \in \{1,\ldots m\}\}$ of preference lists.
\end{lemma}

\begin{proof}
Let $\pi$ be an optimal matching for $B+t$ and $A$. 
If $\pi$ is not efficient for~$L$, there is a better matching. 
That is, there exists a matching $\sigma \neq \pi$ such that for all $b_i \in B$ either 
$\sigma(b_i)=\pi(b_i)$ or $\sigma(b_i)$ appears before $\pi(b_i)$ in $L_t(b_i)$. 
Then $\| b_i+ t-\sigma(b_i) \|^2 \leq \| b_i+ t-\pi(b_i) \|^2$ for all $b_i \in B$ and, since  $\sigma \neq \pi$, 
we have $f(\sigma,t) < f(\pi,t)$, which contradicts the optimality of $\pi$.  \qed
\end{proof}

Note also that the proofs of parts
Lemma~\ref{mainlemma}(\ref{lem2_OnlyMRelevant})--(\ref{cor2_ordereing}) can be carried out in this abstract setting, and
hold for any efficient matching where there are no ties in the preference lists.
Part~(\ref{cor2_ordereing}) immediately yields an upper bound of $m!$ on the number of efficient matchings and, in addition, implies that only the first $m$ elements of each $L_t(b_i)$ are relevant.   
A similar bound for the balanced case is implicitly implied by the results in~\cite{Abdul}.
The bound is tight for the combinatorial problem, since if the ordered lists all coincide there are $m!$ different efficient matchings. 
 
A recent study, motivated by the extended abstract~\cite{HJK},
the precursor of this work, considers this combinatorial problem 
and derives the following result. 

\begin{lemma}[Asinowski et al.~\cite{asinetal}]\label{thm:stable}
The number of elements that belong to some efficient matching 
with respect to $m$ ordered preference lists is at most $m(\ln m +1)$.
\end{lemma}
The properties derived so far imply the following significantly improved upper bound on the complexity of $\D_{B,A}$.

\begin{theorem}\label{thm:ourThm}
The combinatorial complexity of $\D_{B,A}$ is $O( n^2m^{3.5} (e \ln m+e)^m)$.
\end{theorem}
\begin{proof}
The proof has two parts. First, we identify a convex subdivision $K$ such that in each of its 
regions the first $m$ elements of each of the ordered preference lists $L_t(b)$ of neighbors of each $b+t$, according to their distance 
from $b+t$, are fixed for all $b\in B$, and appear in a fixed order in the list. 
We show that the complexity of $K$ is only polynomial;
specifically, it is $O(n^2m^4)$.
Second, we give an upper bound on how many regions of $\D_{B,A}$ can intersect a given region of $K$,
using Lemma~\ref{thm:stable}. Together, these imply an upper bound on the complexity of $\D_{B,A}$.

In order to bound the complexity of $K$, fix $b\in B$, and consider the coarser subdivision $\VD(b,A)$, 
in which only the single list $L_t(b)$ is required to be fixed within each cell.
A naive way of bounding the complexity of $\VD(b,A)$ is to draw all the 
$O(n^2)$ bisectors between the pairs of points in $A-b$, and form their 
arrangement. Each cell of the arrangement has the desired property, as 
is easily checked. As a matter of fact, this naive analysis can be applied 
to the entire structure, over all $b\in B$, in which the arrangement obtained for the individual members $b \in B$ are overlayed into one common arrangement. 
Altogether there are $O(n^2m)$ 
such bisectors, and their arrangement thus consists of $O(n^4m^2)$ regions.

We obtain an improved bound of $O(n^2m^4)$ on the complexity of $K$.
For this, note that it suffices
to draw only relevant portions of some of the bisectors. Specifically, let $b\in B$ 
and $a,a'\in A$. In view of Lemma~\ref{mainlemma}(\ref{lem2_OnlyMRelevant}),
we need to consider only the portion of the bisector $\beta_{a-b,a'-b}$
between $a-b$ and $a'-b$ that consists of those points $t$ such that $a$ and $a'$ 
are among the $m$ nearest neighbors of $b+t$ in $A$; other portions of the bisector 
are ``transparent'' and have no effect on the structure of~$K$. 

In general, the relevant portion of a bisector $\beta_{a-b,a'-b}$ need not be connected.
To simplify the analysis, we will bound the number of (entire) bisectors of this form
whose relevant portion is nonempty. Moreover, we will carry out this
analysis for each $b\in B$ separately.

This analysis can be carried out via the Clarkson-Shor technique~\cite{ClarksonShor}, albeit
in a somewhat non-standard manner. Specifically, with $b$ fixed, we have 
the set $A$ of $n$ points, and a system of bisectors $\beta_{a-b,a'-b}$, 
each defined by two points $a,a'\in A$.  Each bisector $\beta_{a-b,a'-b}$
has a \emph{conflict set}, which we define to be a smallest subset $A'$ of 
$A$, such that there exists a point $t$ on the bisector, such that the two
nearest neighbors of $b+t$ in $A\setminus A'$ are $a$ and $a'$. Clearly, the
conflict set is not uniquely defined, but this is fine for the 
Clarkson-Shor technique to apply, because, if we draw a random sample $R$ 
of $A$, it still holds that the probability that $\beta_{a-b,a'-b}$
will generate an edge of the Voronoi diagram of $R-b$ is \emph{at least} 
the probability that $a$ and $a'$ are chosen in $R$ and none of the points 
in the specific conflict set is chosen. This lower bound suffices for 
the Clarkson-Shor technique to apply, and it implies that the number of 
bisectors that contribute a portion to $\VD(b,A)$ (each of which has a 
conflict set of size at most $m$) is $O(m^2)$ times the complexity of 
the Voronoi diagram of $R-b$, for a random sample $R$ of $n/m$ points of $A$. 
That is, the number of such bisectors is $O(nm)$, instead of the number 
$O(n^2)$ of all bisectors. 
Summing over $b \in B$, we obtain a total of $O(nm^2)$ bisectors instead of $O(n^2m)$. 
The claim about the complexity of $K$ is now
immediate. 

We now consider all possible translations $t$ in the interior of some fixed region 
$\tau$ of $K$ and their corresponding optimal matchings. 
Lemma~\ref{lemmaStabPref}
ensures that all of them must be efficient with respect to the fixed preference lists 
$L_t(b)$, for $b\in B$.
In addition, Lemma~\ref{lem4_sameSize1Match} ensures that we only need to bound the number of different image sets of such efficient matchings. 
Using the bound in Lemma~\ref{thm:stable}, we can derive that the number of optimal matchings for translations in $\tau$ is then at most 
\[ \binom{m(\ln m +1)}{m}  \le \frac{m^m  (\ln m +1)^m}{m!}  = O\left( \frac{ (e\ln m + e)^m }{ \sqrt{m}}  \right),\] where in the second step we used Stirling's approximation. 
Hence, by multiplying this bound by the number of regions in $K$, we conclude that the number of 
assignments corresponding to optimal matchings, and thus also the complexity of 
$\D_{B,A}$, is at most $O(n^2m^{3.5} (e \ln m + e)^m)$. \qed
\end{proof}

The following proposition sets
an obstruction for the combinatorial approach alone to yield a polynomial bound for $\D_{B,A}$.

\begin{proposition}\label{claim:2}
For every $n\geq \lfloor\frac{m}2\rfloor +m $, there exist $m$ preference lists,	 with elements in $\{1,\ldots,n\}$, with $\Omega\left(\frac{2^m}{\sqrt{m}}\right)$ different images of efficient matchings.
\end{proposition}

\begin{proof}
We construct a set of lists such that for every $i \in \{1,\ldots,m\}$ the $\lfloor\frac{m}2\rfloor$ smallest elements are the same (and appearing in the same order); we denote by $S$ the set of these elements. 
For the position $\lfloor \frac{m}2\rfloor+1$ of the lists, we use a set $S'$ of $m$ distinct elements such that $S \cap S'= \emptyset$. 
Given a permutation~$\lambda$ of $\{1,\ldots,m\}$, consider the matching assigning to each $i \in \{1,\ldots,m\}$ the first element in its list,
in the order $\lambda$, that was not assigned to any previous
element. 
It is easy to see that this matching is efficient and that its image consists of $S$ and the subset of $S'$ corresponding to the
last $\lceil \frac{m}2\rceil$ positions of~$\lambda$. 
Therefore, every subset of $S'$ of size $\lceil \frac{m}2\rceil$ is, together with $S$, the image
of an efficient matching. Hence, $\binom{m}{\lceil \frac{m}2\rceil} = \Omega\left(\frac{2^m}{\sqrt{m}}\right)$ 
different sets correspond to images of efficient matchings. \qed
\end{proof}

We now derive a lower bound on the complexity of $\D_{B,A}$.
Consider the arrangement $K$ introduced in the proof of Theorem~\ref{thm:ourThm}, and note that 
in the interior of each of its two-dimensional faces the sequence of the first $m$ points of $A$ closest to 
$b_i$ is fixed, and uniquely defined, for all $b_i \in B$. 
Any optimal matching in the interior of a two-dimensional face of $K$ must be efficient for the corresponding preference lists.  
We will provide a pair of point sets generating many different preference lists that, in addition, have disjoint sets of efficient matchings. 
In order to prove it, we need first the following property of the efficient matchings of a certain type of preference lists.  

\begin{lemma}\label{blockpreference}
Let $\mathcal L$ be a set of preference lists $L_t(b_i)$, for $b_i \in B$, and let $\{a_1,\ldots, a_m\}$ 
be a subset of $m$ distinct points of $A$, satisfying the following conditions, 
for some $0\le j\le m-1$.
\begin{enumerate}[(a)]
\item $L_t(b_i)$ starts with $(a_1,a_2,\dots,a_j)$, for $i=1,\ldots,j$. \label{block1}
\item $L_t(b_{j+1})$ starts with $a_{j+1}$. \label{block2}
\item $L_t(b_i)$ starts with $(a_{j+2},\dots,a_m)$, for $i=j+2,\ldots,m$.  \label{block3}
\end{enumerate}
Then, every efficient matching for $\mathcal L$, matches $B$ to $\{a_1,\ldots,a_m\}$.
\end{lemma}

\begin{proof}
Let $\pi:B\to A$ be an efficient matching for $\mathcal L$. 
We prove the stronger claim that 
\begin{align}
\pi(&b_i)\in\{a_1,\ldots,a_j\},\quad~~\text{for}~i=1,\ldots,j,\label{orit}\\
\pi(&b_{j+1})=a_{j+1},\quad\quad\quad~~\text{and}\nonumber\\
\pi(&b_i)\in\{a_{j+2},\ldots,a_m\},~\text{for}~i=j+2,\ldots, m.\nonumber
\end{align}
For this, we apply the adaptation of Lemma~\ref{mainlemma}(\ref{cor2_ordereing}) to the abstract setting,
by which there exists an ordering $(b_{i_1},\ldots,b_{i_m})$ of the elements of $B$, such that 
each $b_{i_k}$ is assigned by $\pi$ to the most preferred element in $L_t(b_{i_k})$ that is not already claimed 
by one of $b_{i_1},\ldots, b_{i_{k-1}}$, the elements preceding $b_{i_k}$ in this ordering.

Consider the minimal index $k_0$ for which one of the assertions in \eqref{orit} is violated.
If $i_{k_0}=j+1$, then, by the property of the ordering, there exists another element $b_{i_k}$, with $k<k_0$, such that $\pi(b_{i_k})=a_{j+1}$. 
This violates \eqref{orit}.
This contradicts the minimality of $k_0$, so $k_0\neq j+1$.

Similarly, assume $1\le i_{k_0}\le j$. By the property of the ordering, it must be that $k_0> j$, and there exist $j$ elements 
$b_{i_{k_1}},\ldots,b_{i_{k_j}}$, with $k_1,\ldots,k_j<k_0$, such that $\pi(b_{i_{k_l}})=a_l$, for $l=1,\ldots,j$. 
Since $1\le i_{k_0}\le j$, necessarily one of $k_1,\ldots, k_l$ is larger than $j$. This means that \eqref{orit} is violated for some $i_k<k_0$, which again leads to a contradiction.
The case where $k_0>j$ is symmetric and can be treated in the same manner. \qed
\end{proof}

\begin{figure}[t]
\centering
\subfigure[Point sets on the line. ]{\includegraphics[page=4]{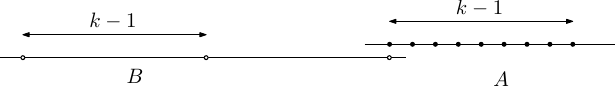}\label{fig:lowerboundexample1}}
\subfigure[Point sets in the plane.]{\includegraphics[page=3]{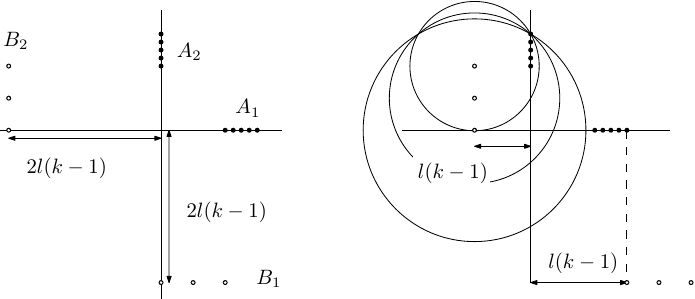}\label{fig:lowerboundexample}}
\caption{Construction of the lower bound example of Theorem~\ref{prop:lowerLS}.}
\end{figure}	

\begin{theorem}\label{prop:lowerLS}
For any $m,n \in \N$ with $n \ge m \ge 2$, there exist planar point sets $A,B$ with $|B|=m$ and $|A|=n$ such that its partial-matching diagram $\D_{B,A}$ has ${\Omega(m^2 (n-m)^2)}$ regions.
\end{theorem}

\begin{proof}
We describe first a construction similar to the one used for the lower bound in Rote~\cite{roteNote1}. 
Let $l<k$ be two integer parameters.
Let~$A$ be the set of $k$ points on the line with coordinates 
$$
(l-1)(k-1),(l-1)(k-1)+1,\dots,l(k-1),
$$ 
and let $B$ be the set of $l$ points with coordinates 
$$0,k-1,2(k-1),\dots,(l-1)(k-1),$$
 as depicted in Figure~\ref{fig:lowerboundexample1}. 
Enumerate the points of $A$ as $a_1,\ldots, a_k$ and those of $B$ as $b_1,\ldots, b_l$ in their left-to-right order.
Note that, for every point in $B$, the $l$ closest points of $A$ are the~$l$ leftmost points $a_1,\ldots,a_l$ (in this order).
We start moving $B$ to the right.
During the motion, exactly one point $b_j$ of $B$ lies ``inside'' $A$ 
(until all the points of $B$ have crossed $A$), and 
$b_j$ traverses all the Voronoi regions of $a_1,\ldots, a_k$ in this order.
Fix $1\le j\le l-1$, $j+1\le i\le k-l+j+1$, 
and consider the instance of the motion 
when $b_{j+1}$ crosses the Voronoi region of $a_i$.
At this translation $t$, the $j$ closest neighbors in $A$
of every $b_{j'}$, ${j'}\le j$, are the $j$ leftmost points $a_1,\ldots, a_j$, and the $l-j-1$
closest neighbors of every $b_{j'}$, ${j'}\ge j+2$, are the $l-j-1$
rightmost points $a_{k-l+j+2},\ldots,a_k$. 
By Lemma~\ref{blockpreference}, in an efficient matching of $B+t$ to $A$,
the matched subset of $A$ is $\{a_1,\ldots, a_j,a_i,a_{k-l+j+2},\ldots,a_k\}$.
We thus obtain at least $l(k-l)$ distinct optimal matchings.

In order to construct the two-dimensional instance, we assume for simplicity that $m=2l$ and $n=2k$ 
are even, and construct the following copies of the sets $A$ and $B$. 
\begin{align*}
A_1 &= \{ (lk-l-k+t,0) \mid t \in \{1,\ldots,k\}  \},   \\
 A_2 &= \{ (0,lk-l-k+t) \mid t \in \{1,\ldots,k\}  \},   \\
B_1 &= \{ ((t-1)(k-1),-2l(k-1))\mid t \in \{1,\ldots,l\}  \} ,\\
B_2 &= \{ (-2l(k-1),(t-1)(k-1))\mid t \in \{1,\ldots,l\}  \} ;
\end{align*}
see Figure~\ref{fig:lowerboundexample}.
Set $A=A_1\cup A_2$ and $B=B_1\cup B_2$.
For $s=1,2$, enumerate the elements of $A_s$ as $a_1^{(s)},\ldots,a_k^{(s)}$,
and those of $B_s$ as $b_1^{(s)},\ldots,b_l^{(s)}$, in increasing order of the varying coordinate.
We claim that for every choice of 
\begin{align*}
1&\le j_1,j_2\le l-1\\
j_1+1&\le i_1\le k-l+j_1+1,\\
j_2+1&\le i_2\le k-l+j_2+1,
\end{align*}
there exists a translation $t=(t_1,t_2)$, at which the elements of $B_s+t$ are matched
to the subset $\{a_1^{(s)},\ldots,a_{j_s}^{(s)},a_{i_s}^{(s)},a_{k-l+j_s+2}^{(s)},\ldots,a_k^{(s)}\}$
of $A_s$, for $s=1,2$.
Indeed, take $t_1$ to be any horizontal translation at which $B_1$ is matched to
$\{a_1^{(1)},\ldots,a_{j_1}^{(1)},a_{i_1}^{(1)},a_{k-l+j_1+2}^{(1)},\ldots,a_k^{(1)}\}$,
as provided in the one-dimensional construction.
(Note that the fact that $B_1$ and $A_1$ are not collinear does not matter,
because the order of distances in $(B_1+t_1)\times A_1$ is the same as the order that would arise if we projected
$B_1+t_1$ onto the line ($x$-axis) containing $A_1$.)
We define $t_2$ in a fully symmetric manner. The claim holds because, for any such $t$, the distance
between any point $b\in B_1+t$ to any point of $A_2$ is larger than any of the distances
between $b$ and the points of $A_1$, and, symmetrically, the distance between any $b\in B_2+t$
to any point of $A_1$ is larger than any of its distances to the points of $A_2$.
Indeed, the largest translation of $B_1$ to the right,
until (the $x$-projection of) all its points cross $A_1$ is $l(k-1)$, and a 
symmetric claim holds for $B_2$ and $A_2$.
Hence, for any relevant $t=(t_1,t_2)$, the points of $B_2$ are at distance $2l(k-1)-t_1$
to the left of  the $y$-axis, so their distances from the points of $A_1$ are at least 
$$
2l(k-1)-t_1+l(k-1)-(k-1)=(3l-1)(k-1)-t_1;
$$
see the right part of Figure~\ref{fig:lowerboundexample}.
On the other hand, the largest distance of a point of $B_2+t$ from the points of $A_2$ is 
at most (again, see Figure~\ref{fig:lowerboundexample}) $\sqrt{(2l(k-1)-t_1)^2+(l(k-1))^2}$.
Since $t_1\le l(k-1)$, straightforward calculation shows that
$\sqrt{(2l(k-1)-t_1)^2+(l(k-1))^2}< (3l-1)(k-1)-t_1$,
provided that $l$ is at least some small absolute constant.
This establishes the claim, and shows that the number of efficient matchings in this case is at least 
$$
(l(k-l))^2=\frac{1}{4}m^2(n-m)^2,
$$
as asserted.
The example can be easily perturbed such that $A \cup B$ is in general position. \qed
\end{proof}

\section{The partial-matching RMS distance under translation} 
\label{sec:pmalg}

We now concentrate on the algorithmic problem of computing, in polynomial time,
a local minimum of the partial-matching RMS distance under translation.
Before going into the implementation details, we describe the main ideas of the algorithm.
\subsection{ \bf The high-level algorithm}
We ``home in'' on a local minimum of $F(t)$ by maintaining a vertical slab $I$ in the plane that is known to contain such
a local minimum in its interior, and by repeatedly shrinking it until we obtain a slab $I^*$
that does not contain any vertex of $\D_{B,A}$. That is, any (vertical) line contained in $I^*$
intersects the same sequence of regions, and, by Theorem~\ref{thm:Rote}, the number
of these regions is $O(nm)$. 
We then find an optimal partial matching assignment in each region, applying the Hungarian algorithm
described in the next subsection, and the corresponding explicit (quadratic)
expression of $F(t)$, and search for a local minimum within each region. 

A major component of the algorithm is a procedure, that we call $\Pi_1(\ell)$,
which, for a given input line $\ell$, constructs the intersection of $\D_{B,A}$ with $\ell$, computes a 
global minimum $t^*$ of $F$ on $\ell$, and determines a side of $\ell$, in which
$F$ attains strictly smaller values than $F(t^*)$. 
If no such decrease is found
in the neighborhood of $t^*$ then it is a local minimum of $F$, and we stop.

We use this ``decision procedure'' as follows.
Suppose we have a current vertical slab $I$, bounded on the left by a line 
$\ell^-$ and on the right by a line $\ell^+$. We assume that $\Pi_1$ has 
been executed on $\ell^-$ and on $\ell^+$, and that we have determined that 
$F$ assumes smaller values than its global minimum on $\ell^-$ to the 
right of $\ell^-$, and that it assumes smaller values than its global minimum 
on $\ell^+$ to the left of $\ell^+$. 
As we argue below, this implies that $F$ has a local minimum in the interior of $I$. 
Let $\ell$ be some vertical line passing through~$I$. We run $\Pi_1$
on $\ell$. If it determines that $F$ attains smaller values to its left
(resp., to its right), we shrink $I$ to the slab bounded by $\ell^-$ and $\ell$
(resp., the slab bounded by $\ell$ and $\ell^+$).  
By what will be argued below, this ensures that the new slab also contains a local minimum of $F$ in its interior.

We note that by restricting the problem to a line $\ell$ in the decision procedure described above we face a one-dimensional version of the problem. 
However, here it does not suffice to find a \emph{local} minimum over $\ell$. 
To see why, consider the following situation (as illustrated in Figure~\ref{fig:fig_localvsglobal}): Let $I$ be a slab of the form $x_1 \leq x \leq x_2$ in the plane, and let $q_1$ and $q_2$ be local minima of $F$ on $\{(x_1,s) \mid s\in \R \}$ (with $\frac{\partial F}{\partial x}(q_1)<0$) and on $\{(x_2,s) \mid s\in \R \}$ (with $\frac{\partial F}{\partial x}(q_2)>0$), respectively. 
Then we cannot conclude that~$I$ contains a local minimum. Indeed, it might be the case that the negative slope (in the $x$-direction) at $q_1$ points to a minimal point to the right of $x_2$, and the positive slope at $q_2$ points to a minimal point to the left of $x_1$, while the slab itself does not contain any local minimum. 
However, if $q_1$ and $q_2$ are \emph{global} minima, and the signs of the derivatives are as before, a local minimum must exist within the slab.
Indeed, 
it is easily checked that $I$ contains a minimum point of $F$ (restricted to $I$), simply because $F$ tends to $+ \infty$ as $|y| \to \infty$.
The conditions on the derivatives at $q_1,q_2$ indicate that this minimum is not attained on the boundary of $I$, and thus it must be a local minimum.
The description so far has assumed that the slab is bounded on the left and on the right, but the argument works equally well for semi-unbounded slabs. 

Nothing has so far been said about the concrete choice of the ``middle'' line~$\ell$. 
This will be spelled out in the detailed description of the algorithm, which now follows.

 \begin{figure}[ht]
 \begin{center}
 \includegraphics{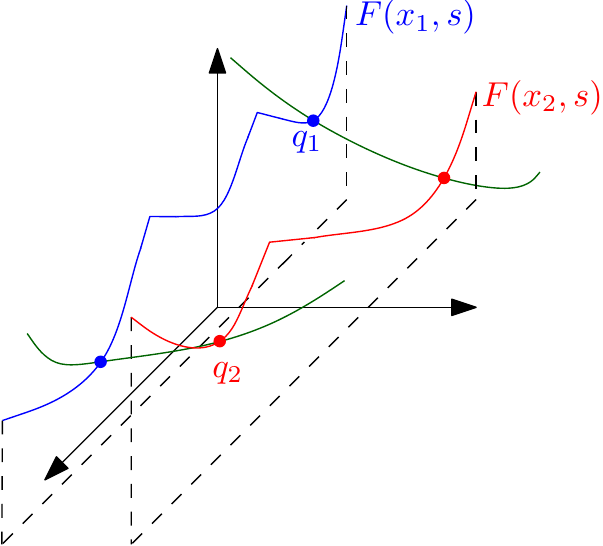}
 \caption{The function $F$ restricted to a slab. The graphs of $F$ over the two lines bounding the slab are highlighted, with two boundary \emph{local} minima $q_1$ and $q_2$ satisfying $\frac{\partial F}{\partial x}(q_1)<0$, $\frac{\partial F}{\partial x}(q_2)>0$. However, there is no local minimum inside the slab.}
 \label{fig:fig_localvsglobal}
 \end{center}
 \end{figure}

To initialize the slab, we choose an arbitrary \emph{horizontal} line $\lambda$,
and run $\Pi_1$ on $\lambda$, to find the sequence $S$ of its intersection 
points with the edges of $\D_{B,A}$.  We run a binary search through $S$, 
where at each step we execute $\Pi_1$ on the vertical line through the 
current point. When the search terminates, we have a vertical slab $I_0$ 
whose intersection with $\lambda$ is contained in a single region~$\sigma_0$ 
of~$\D_{B,A}$.  
(Note that $I_0$ might be semi-unbounded, if it lies to the left of the leftmost intersection, or to the right of the rightmost intersection, of $\lambda$ with the edges of~$\D_{B,A}$.)

\begin{figure}[ht]
\begin{center}
\includegraphics{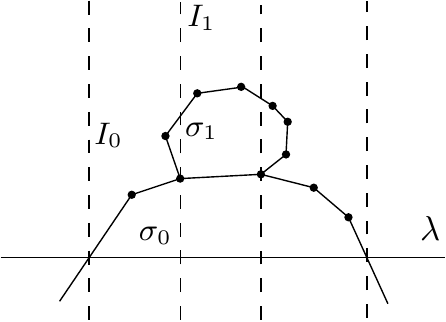}

	\caption{Shrinking the slab from $I_0$ to $I_1$.}
	\label{fig:fig_searchCell}
\end{center}
\end{figure}

After this initialization, we find  the region $\sigma_1$ that lies directly above $\sigma_0$ and 
that the final slab $I^*$ should cross\footnote{Since we only seek a local minimum of $F$, $I^*$ is not unique in general. When we speak of \emph{the} final slab, we simply mean the one produced by our procedure.}. 
In general, there are possibly many regions that lie above $\sigma_0$, but fortunately, by Lemma~\ref{lem:structure}(\ref{lem_bound_edges}),
their number is only at most $m(n-m)$.

To find $\sigma_1$, we compute the boundary of $\sigma_0$; this is done similarly to the execution of $\Pi_1$ (see details in Subsection~\ref{sub:boundary}). 
Once we have explored the boundary of $\sigma_0$, we take the sequence of all vertices of $\sigma_0$,
and run a $\Pi_1$-guided binary search on the vertical lines passing through them, exactly as we did with
the vertices of $S$, to shrink $I_0$ into a slab $I_1$,
so that the top part of the intersection of $\sigma_0$ with $I_1$ is a (portion of a) single edge.  
This allows us to determine $\sigma_1$,
which is the region lying on the other (higher) side of this edge.
See Figure~\ref{fig:fig_searchCell} for an illustration.
A symmetric variant of this procedure will find the region lying directly below
$\sigma_0$ in the final slab.

We repeat the previous step to find the entire stack of $O(nm)$ regions that~$I^*$ crosses,
where each step shrinks the current slab and then crosses to the next region in 
the stack. Once this is completed, we find a local minimum within~$I^*$ as
explained above. 
(The reader should keep in mind that the procedure can, and will, stop at any time when the line on which $\Pi_1$ is run is found to contain a local minimum of $F$.)

\subsection{ \bf Partial matching at a fixed translation: The Hungarian algorithm}
The Hungarian method, developed by Kuhn in 1955 \cite{Kuh55}, is an efficient procedure
for computing a perfect maximum weight (or, for us, minimum weight) bipartite matching between
two sets $A, B$ of equal size $m$, with running time $O(m^4)$, which has been improved to $O(m^3)$ by Edmond and Karp \cite{EK}. 
The original algorithm proceeds iteratively, starting with an empty set $M_0$ of matched pairs. In the $i$-th iteration it takes the
current set $M_{i-1}$ of $i-1$ matched pairs, and transforms it into a set $M_i$ with $i$
matched pairs, until it obtains the desired optimal perfect matching with $m$ pairs.

Let us sketch the technique for  minimum-weight matching, which is the one
we want. The $i$-th iteration is implemented as follows. Define $D$ to be the (bipartite) 
\emph{directed} graph, with vertex set $A\cup B$, whose edges are the edges of $M_{i-1}$ 
directed from $B$ to $A$, and the edges of $(A\times B)\setminus M_{i-1}$, directed 
from $A$ to $B$. We look for an augmenting path $p$ that starts at $B$ and ends at $A$,
of minimum weight, and we set $M_i := M_{i-1} \triangle p$ (here $\triangle$ denotes symmetric difference).

Ramshaw and Tarjan~\cite{tarjan} proposed and analyzed an adaptation of the Hungarian Method to unbalanced bipartite graphs. 
They used a modification of Dijkstra's algorithm, as described in \cite{Fredman}, for finding the augmenting paths. 
The analysis of the careful implementation of the Hungarian method that they propose yields a running time of $O(n m^2)$ for graphs with vertex sets of sizes $m$ and $n$, respectively, assuming $m \leq n$.

Hence, to recap, given a translation $t$, we can compute $M(B+t,A)$ by the above
algorithm, where the weight of an edge $(a,b)\in A\times B$ is $\|b+t-a\|^2$.
We denote this procedure as $\Pi_0(t)$; its output is the set of matched pairs,
or, in our notation, the injective assignment $\pi: B \to A$.

\subsection{\bf Computing the boundary of $\sigma_0$} \label{sub:boundary}
Let $\sigma_0$ be an open region of $\D_{B,A}$, and let $A_0 \subseteq A$ be the set of the 
$m$ matched points of $A$, for translations $t\in \sigma_0$. 
$A_0$ can be computed by picking some translation $t_0$ (interior to) $\sigma_0$, and then 
by running $\Pi_0(t_0)$, for finding an optimal matching $M_0$ for the translation $t_0$ in time $O(nm^2)$.
By Lemma~\ref{lemdirections}, we know that there are $O(nm)$ possible directions for the bisectors forming $\bd\sigma_0$. 
Moreover, when we cross an edge of $\sigma_0$, an optimal matching $M_1$ that replaces $M_0$ is obtained from a collection of pairwise-disjoint alternating paths (and possibly also cycles), where in each path we replace the edges of $M_0$ in the path by the (same number of) edges of~$M_1$. 
As seen in the proof of Lemma~\ref{lemdirections}, if $\gamma$ is a cycle, then $d_{\gamma}=0$ and therefore whether it increases, preserves or decreases the cost of a matching is independent of the translation $t$. 
Thus, we can assume that the symmetric difference of $M_0$ and $M_1$ consists only of paths by flipping in $M_1$ the cycles in the difference (obtaining a matching with the same cost as $M_1$ everywhere).
The subset $A_1$ of the $m$ matched points of $A$ in $M_1$ is obtained by replacing, for each of these paths, the starting point $a_i$ of the path (which belongs to $A_0$) by the terminal point $a_j$ (which belongs to $A_1$). 
As shown in  Lemma~\ref{lemdirections}, this implies 
that the bisector through which we have crossed from $\sigma_0$ to the neighbor region~$\sigma_1$ must be perpendicular to $a_i-a_j$, for all pairs $(a_i,a_j)$.
Moreover, assuming general position, and specifically that there are no two distinct pairs of points $\{a_p,a_r\},\{ a_q,a_s\} \subset A$ such that $a_r-a_p$ and $a_s-a_q$ are parallel, it follows that each edge of $\DD$, and specifically of $\bd\sigma_0$, corresponds to a \emph{single} such alternating path, and to a single replacement pair $(a_i,a_j)$. 
In other words, under the above general position assumption, over each edge of $\sigma_0$ only one point $a_i\in{A_0}$ exits the optimal matching and another point $a_j \in{A \setminus A_0}$ enters in it. 
(Note however that the edges of the matching can change globally.)
Therefore, we can construct $\bd\sigma_0$ easily and efficiently in the following manner. For each of the $m$ points $a_i\in{A_0}$, we replace it by one of the $n-m$ points $a_j\in{A \setminus A_0}$. For each such replacement we compute the new optimal (perfect) matching $M_1$ between $B$ and 
$A_1=(A_0\setminus \{a_i\})\cup \{a_j\}$ (recall that, by Lemma~\ref{lem4_sameSize1Match}, once $A_1$ is fixed, the matching 
$M_1$ is independent of the translation, so it can be computed at any translation, e.g., at $t_0$). 
We then find the bisector, by comparing the expression in the right-hand side of (\ref{eq:Rote}) for the new matching $M_1$ and for the optimal matching $M_0$ in $\sigma_0$. This provides us with a total of $O(nm)$ potential bisectors. 
We now obtain $\sigma_0$ as the intersection of the $O(nm)$ halfplanes, bounded by these bisectors and containing $t_0$. 
This takes $O(nm\log{(nm)})$ additional time, which is dominated by the time used for the computation of the optimal matchings in $\sigma_0$ and across its potential edges.
Note that for each edge on $\bd\sigma_0$ we also know an optimal assignment on its other side. 

If the points are not in general position, it is not obvious that the edges of $\sigma_0$ can be constructed using the same procedure. 
The difference here is that the set matched in a neighboring region might differ in more than one point from~$A_0$. 
This happens only if more than one (vertex independent) paths in the symmetric difference of the corresponding matchings vanish on the same line $\mu$. 
Nevertheless, we are fine in such a case, because each of the above paths could be flipped independently, inducing a valid matching, with the same cost as $\sigma_0$ on $\mu$, whose matched set differs from $A_0$ by only one element. 
More precisely, we have that if $A_1 \subseteq A$ is the set matched by an optimal matching $M_1$ in a region $\sigma_1$ sharing an edge $e$ with $\sigma_0$, then there is at least one pair of points  $a_j \in A_1 \setminus A_0$ and  $a_i \in A_0 \setminus A_1$ such that the bisector between~$\sigma_0$ and the best matching $\tau$ using ($A_0 \setminus a_i) \cup a_j$ supports $e$.    
Since the bisector of $\sigma_0$ with $\tau$ is one of the potential bisectors we construct, each edge of $ \sigma_0$ is discovered by the procedure and, hence, the computed boundary of $\sigma_0$ is correct. 

\subsection{\bf Computing the optimal matching beyond an edge}

As argued above, if the point sets are in general position, we can easily compute an optimal matching for a neighboring region by flipping, in the current optimal matching, the alternating path inducing the common edge. 
However, in degenerate situations, we cannot directly infer any optimal matching on the other side if several paths induce the same edge.
To compute the new matching, we can apply the Hungarian algorithm to a translation infinitesimally beyond the edge. 
That is, we take the midpoint $\bar{e}$ of the edge $e$ and its outer normal vector $s$ with respect to $\sigma_0$, and compute an optimal matching for a translation $\bar{e}+\varepsilon s$ for $\varepsilon >0$ arbitrarily small. 
In order to do it, we compute the sums and perform the comparisons required by the algorithm regarding the squared distances $\|b+\bar{e}+\varepsilon s - a \|^2$, for every $a \in A$ and $b \in B$, as polynomials of degree two in $\varepsilon$, evaluated in the limit $\varepsilon\,{\scriptstyle\searrow}\,0$.  

Fortunately, the overhead incurred by these additional computations is dominated by the running time of our procedures for points in general position, as the following analysis shows.

\subsection{\bf Solving $\Pi_1(\ell)$}
Let $\ell$ be a given line in $\R^2$; without loss of generality assume $\ell$ to be vertical.
We start at some arbitrary point $t_0\in\ell$,
run $\Pi_0(t_0)$, and obtain an optimal injective assignment $\pi_0$ for the partial 
matching between $B+t_0$ and~$A$. We now proceed from $t_0$ upwards along $\ell$, and 
seek the intersection of this ray with the boundary of the region $\sigma_0$ of $\D_{B,A}$
that contains $t_0$. 
Finding this intersection will also identify the next region of the subdivision
that $\ell$ crosses into (as noted above, this identification is cheap in general position, but requires some work in degenerate cases), and we will continue in this manner, finding all the regions of $\D_{B,A}$
that the upper ray of $\ell$ crosses. 
In a fully symmetric manner, we find the regions
crossed by the lower ray from $t_0$, altogether $O(nm)$ regions, by Theorem~\ref{thm:Rote}.

To find the intersection $t^*$ of the upper ray of $\ell$ with $\bd\sigma_0$, we apply a simplified variant of the procedure for computing $\bd\sigma_0$. That is, we construct the $O(nm)$ potential bisectors between $\sigma_0$ and the neighboring regions, exactly as before. 
(Note that, as argued above for constructing $\partial \sigma_0$, these bisectors determine the boundary of the region even in degenerate cases.)
The point~$t^*$ is then the lowest point of intersection of $\ell$ with all these bisectors lying above~$t_0$. We repeat this process for each new region that we encounter, and do the same
in the opposite direction, along the lower ray from $t_0$, until we find all the regions
of $\D_{B,A}$ crossed by $\ell$.

The number of regions is $O(nm)$. We compute the explicit expression for $F(t)$ in each 
of them, and thereby find the global minimum $\bar{t}$ of $F$ along $\ell$. Finally, we 
compute $\frac{\bd F}{\bd x}(\bar{t})$ (which is a linear expression in $t$, readily obtained
from the explicit quadratic expression for $F$ in the neighborhood of $\bar{t}$). We note that $\bar{t}$ cannot be a breakpoint of $F$ (that is, lie on an edge of $\DD$), since a local minimum in $\bar{t}$ implies that $F(t)$ in both neighboring regions is decreasing towards $\bar{t}$, but no bisecting edge can pass through such a point.
If $\frac{\bd F}{\bd x}(\bar{t})$ is negative (resp., positive), we conclude that $F$ attains lower values than 
its minimum on $\ell$ to the right (resp., left) of $\ell$, and we report this direction.
If the derivative is~$0$, we have found a local minimum of $F$ and we stop the whole algorithm.

The cost of $\Pi_1(\ell)$ is $O(nm \cdot nm \cdot m^3)=O(n^2m^5)$, as we encounter $O(nm)$ regions along $\ell$, and for each of them we examine $O(nm)$ potential bisectors, each of which is obtained by running $\Pi_0$, in $O(m^3)$ time. 
The additional time to construct an optimal matching in each region if the points are not in general position is $O(nm^2)$ per region, for a total of $O(nm^2 \cdot nm)=O(n^2m^3)$, and it is thus dominated by the previous bound. 

\subsection{\bf Running time of the algorithm}\label{subsec:pmalg}
The running time of the whole algorithm is dominated by the cost of constructing the $O(nm)$ regions that the final slab $I^*$ crosses. Each region is constructed in $O(nm^4)$ time, after which we run a $\Pi_1$-guided binary search through its vertices, in time $O(n^2m^5\log{(nm)})$. Multiplying by the number of regions, we get a total running time of $O(n^3m^6\log{(nm)})=O(n^3m^6\log{n})$. 

If the point sets are not in general position, we might need to recompute an optimal matching from scratch when we enter a new region in the final slab. This amounts to $O(nm)$ computations requiring $O(nm^2)$ time each and, thus, it does not increase the total running time of the algorithm.

In summary, we have the following main result of this section.

\begin{theorem}\label{thmMatchingLocMin}
Given two finite point sets $A,B$ in $\R^2$, with $n=|A|>|B|=m$, a local minimum of the partial-matching RMS distance under translation can be computed in $O(n^3m^6\log{n})$ time.
\end{theorem}

\noindent{\em Remark.} 
Returning to the discussion in the introduction concerning heuristics for finding the global minimum,
we can apply similar heuristics to our algorithm.
For example, having an initial translation $t_0$ that we expect (or hope) to be close to the global minimum, 
we can enclose $t_0$ by an initial, sufficiently narrow, slab $I_0$, and run the preceding algorithm starting with $I_0$ 
(after verifying that it contains a local minimum), as described above.
However, if we want to ensure that the global minimum is obtained, we apply the modified approach
in the following subsection.

\subsection{\bf Finding a global minimum}

Using the techniques from the previous subsection, we can easily devise an algorithm for constructing the entire subdivision $\D_{B,A}$, 
by starting at an arbitrary translation, constructing the region that contains it, and continue constructing the neighboring regions. 
Computing an optimal matching in a region $\sigma$ (if it needs to be computed from scratch) takes $O(nm^2)$ time and computing 
the boundary of $\sigma$ takes $O(nm^4)$ time, as explained above. 
In other words, the overall time for constructing $\D_{B,A}$, including an optimal assignment over each of its features,
is $O(nm^4)$ times its complexity, namely $O(n^3m^{7.5}(e\ln m+e)^m)$.

Once $\D_{B,A}$ is constructed, a minimum of $F$ in each region can be computed in time proportional to its complexity, 
since the minimum might be attained at a vertex of the region or at the orthogonal projection of the minimum of the corresponding 
parabola on an edge of the region. 
Therefore, the global minimum of $F$ can be obtained, by traversing all features of $\D_{B,A}$, in time proportional to its complexity. 

Overall, we obtain the following result.
\begin{theorem}\label{global}
Given two finite point sets $A,B$ in $\R^2$, with $n=|A|>|B|=m$, the subdivision $\D_{B,A}$ can be constructed in $O( n^3  m^{7.5}  (e \ln m +e )^m)$ time. 
The global minimum of the partial-matching RMS distance under translation can then be computed in additional $O( n^3  m^{7.5}  (e \ln m +e )^m)$ time. 
\end{theorem}
 In particular, when $m$ is a constant the global minimum (and the entire $\D_{B,A}$)
 can be computed in $O(n^3)$ time.

\section{The Hausdorff RMS distance under translation}
\label{sec:haus}

In this section, we turn to the simpler problem involving the Hausdorff RMS distance,
and present efficient algorithms for computing a local minimum of the RMS function
in one and two dimensions.
We begin with the simpler one-dimensional case.

\subsection{\bf The one-dimensional case.}

\paragraph{The function $N_A$.}
Let $A=\{a_1,\ldots,a_n\}\subset\R$, and recall that, for a point $t\in\R$, $N_A(t)$ denotes a nearest neighbor of $t$ in $A$. 
Then $N_A$ is a step function with the following structure. Assume that the 
elements of $A$ are sorted as $a_1<a_2<\cdots<a_n$, and put $\mu_i=\frac{a_i+a_{i+1}}{2}$, 
for $i=1,\ldots,n-1$. Then, breaking ties in favor of the larger point, we have:
\[
N_A(t) = \left \{ 
\begin{array}{l l}
a_1 & \quad \text{for $t<\mu_1$}\\
a_i & \quad \text{for $i=2,\ldots,n-1$ and $\mu_{i-1} \leq t < \mu_i$}\\
a_n & \quad \text{for $t \geq \mu_{n-1}$.}
\end{array} \right.
\]
See Figure~\ref{fig:fig_Nb} for an illustration. 
We put $N_{A,i}(t):=N_A(b_i+t)$, the nearest 
neighbor in $A$ of $b_i+t$, for $b_i\in{B}$, and $t \in{\R}$, for $i=1,\ldots, m$. 
The graph of each of these $m$ functions is just a copy of the graph of $N_A$, $x$-translated to the left by the respective $b_i$.
\begin{figure}[ht]
\begin{center}
\includegraphics{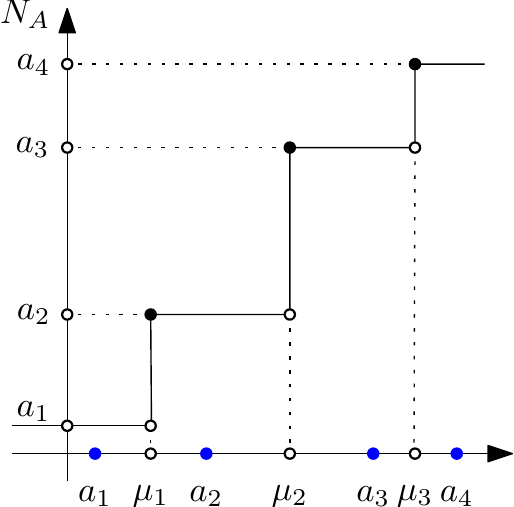}
\caption{The step function $N_{A}$ for a set of points on the line.}
\label{fig:fig_Nb}
\end{center}
\end{figure}

\paragraph{The one-dimensional unidirectional case.}

Denote $\h(B+t,A)$ as $r(t)$ for short, regarding $A$ and $B$ as fixed. 
We thus want to compute a (local) minimum of the function
$$
r(t) = \sum_{i=1}^m (b_{i} + t - N_{A,i}(t))^{2} .
$$
We observe that $r(t)$ is continuous and piecewise differentiable (except at 
the points of discontinuity of the step functions $N_{A,i}(t)$) and 
each of its pieces is a parabolic arc. For any $t$, which is not one of these 
singular points, also referred to as \emph{breakpoints}, the derivative of each of 
the step functions is 0.  Hence we have, for any non-singular local minimum $t$ of $r$,
\begin{equation}\label{eq:HausRMS}
r'(t) = 2\sum_{i=1}^m{( b_i + t - N_{A,i}(t))} = 0 .
\end{equation}
Clearly, for any given (non-singular) translation $t_0$, $r(t_0)$ and $r'(t_0)$ 
can be computed (from scratch) in $O(m\log{n})$ time, provided that the breakpoints
of $N_A$ are given in a sorted order. This will be the case where $A$ is given to us in 
sorted order; otherwise, another $O(n\log n)$ time for sorting $A$ is required.
This also holds for the left and right one-sided derivatives of $r(t_0)$, 
at a breakpoint $t_0$, denoted respectively as $r'(t_0)^-$ and $r'(t_0)^+$.

We note that a local minimum of $r(t)$ cannot occur at a singular point. Indeed, 
for the local minimum to occur at a breakpoint $t_{step}$, we must have 
$r'(t_{step})^- \leq 0$ and $r'(t_{step})^+ \geq 0$. However, referring to 
equation~(\ref{eq:HausRMS}), one easily verifies that the value of $r'(t)$ can 
only decrease at $t_{step}$, contradicting the above inequalities.

Another simple observation is that if there is an interval $I=[t_1,t_2]$, 
such that~$r'(t_1)^+<0$ and $r'(t_2)^->0$, then there exists a local 
minimum of $r(t)$ inside $I$. Our algorithm starts with a large interval with 
this property, and shrinks it repeatedly, while ensuring that it still contains a local minimum. 
At every step of the shrinking process, the number of 
breakpoints of $r(t)$ over $I$ reduces by (at least) half, and the process 
terminates when $I$ does not contain any breakpoints. 
At this point it is straightforward to calculate a local minimum by constructing the explicit representation of $r$ over $I$.

Assuming after relabeling that $b_1<\cdots<b_m$, set $t_1 =\mu_1-b_m$, and $t_2=\mu_{n-1}-b_1$. 
We start with $I=[t_1,t_2]$. It is easily checked that $r(t)$ has no breakpoints outside $I$, and it is 
in fact decreasing for $t<t_1$ and increasing for $t>t_2$. Thus $I$ contains
the global minimum of $r(t)$.

We next describe the procedure for shrinking $I$, i.e., computing a subinterval 
$I' \subset I$, such that the number of breakpoints of $N_{A,i}(t)$ over $I'$ is 
(approximately) half the number of its breakpoints over $I$, for $i=1,\ldots,m$, while maintaining 
the invariant that $r'(t)^+<0$ at the left endpoint of $I'$, and $r'(t)^->0$ 
at the right endpoint. Such a ``halving'' of $I$ is performed in a single 
iteration of the procedure, and since we start with $O(n)$ breakpoints for each function $N_{A,i}(t)$,
the algorithm executes $O(\log n)$ iterations.

The shrinking process is performed as follows.

\begin{enumerate}[(1)]
\item Each iteration starts with an interval $I=[t_1,t_2]$ such that $r'(t_1)^+<0$ 
and $r'(t_2)^->0$ (where the initial values of $t_1$ and $t_2$ are given above). 
We calculate, for each $i=1,\ldots,m$, the median step $\xi_i$ of $N_{A,i}$ among 
its steps within $I$, which can be done in $O(\log{n})$ time.
These $m$ median steps are thus found in total time $O(m\log{n})$. We sort them into a list $L=(\xi_1,\ldots,\xi_m)$.

\item We perform a binary search over $L$ for finding a local minimum of $r(t)$ 
between two consecutive elements of $L$. At each step of the search, with some 
value $t=\xi_k$, we compute $r'(\xi_k)^-$ and $r'(\xi_k)^+$ in $O(m\log{n})$ 
time; as noted earlier, there are only three possible cases:
	\begin{description}
	\item{(a)} If $r'(\xi_k)^->0$ and $r'(\xi_k)^+>0$, we go to the left, replacing $t_2$ by $\xi_k$.
	\item{(b)} If $r'(\xi_k)^-<0$ and $r'(\xi_k)^+<0$, we go to the right, replacing $t_1$ by $\xi_k$.
	\item{(c)} If $r'(\xi_k)^->0$ and $r'(\xi_k)^+<0$, it does not matter where to go---there 
	are local minima on both sides; we go to the left, say, resetting $t_2$ as in~(a).	
	\end{description}
In this way, we maintain our invariant. At the end of the binary search, we get 
an interval $I'=[\xi_j,\xi_{j+1}]$ with $r'(\xi_j)^+<0$ and $r'(\xi_{j+1})^->0$. 
The progress that we have made by passing from $I$ to $I'$ is that, for each step 
function $N_{A,i}$, we got rid of at least half of its steps within $I$: if 
$\xi_k \leq \xi_j$ (resp., $\xi_k \geq \xi_j$) then the leftmost (resp., rightmost)
half of the steps of $N_{A,i}$ within $I$ is discarded.
The running time of this step is $O(m\log n\log m)$ (there are $O(\log m)$ binary 
search steps, each taking $O(m\log n)$ time).

\item We keep shrinking $I$, until it contains no breakpoints (of any $N_{A,i}$). 
We then explicitly construct the graph of $r$ over $I$, and search for a local minimum of $r$ in constant time.
By the invariant, at least one such minimum will be found. 
The overall running time of step (3) is thus $O(m\log n)$.
\end{enumerate}
Hence, the overall running time of the algorithm  is 
$O( n+m\log^2{n}\log m)$.

\paragraph{An improved algorithm.} 
Step (2) of the preceding algorithm performs $O(\log{m})$ binary-search steps over 
the list $L$ of the median breakpoints of the step functions within $I$, in order to 
eliminate (at least) half of the breakpoints inside the interval, but we can get 
rid of a quarter of these breakpoints by replacing the median of $L$ by a {\em weighted} median
and performing just the first step of the search. 
Taking this approach, we keep iterating until the total number of breakpoints (from all $N_{A,i}$) within $I$ is $O(\max\{n,m\})$.
Since we start with $O(nm)$ breakpoints, it is required to perform $\log\min\{n,m\}$ iterations.
Omitting the rather standard details,
we obtain with this improvement the following result.

\begin{theorem} \label{th1dir1dim}
Let $A$, $B$ be two finite sets on the real line, with $|A|=n$ and $|B|=m$, and assume that the 
elements of $A$ are sorted as $a_1<a_2<\cdots<a_n$.
Then a local minimum of the unidirectional RMS distance under translation from $B$ to $A$ 
can be obtained in time 
$$
O\left(\min\left\{ n+m\log^2{n}\log m,~n\log n+m\log n\log{\min{\{n,m\}}}\right\}\right).
$$
\end{theorem}

When $|A|=|B|=n$, the running time is simply $O(n\log^2n)$. When $B$ is much smaller than $A$, namely, when $|A|=n$ and $|B|=O(n/\log^2n)$, the running  
time is $O(n)$.

\paragraph{The one-dimensional bidirectional case.}
Simple extensions of the procedure given above apply to the two variants of
the minimum bidirectional Hausdorff RMS distance, as defined in the introduction. 
One feature to observe is that, for the $L_\infty$-bidirectional case, a local minimum can also arise at an intersection between two one-direction functions $\h(A,B+t)$ and $\h(B+t,A)$.
Nevertheless, it is easy to handle this new kind of breakpoints, and their number is proportional to the number of the other, standard breakpoints. 
Omitting the further fairly routine details of these extensions, we obtain:
\begin{theorem}\label{th2dir1dim}
Given two finite point sets $A$, $B$ on the real line, with $|A|=n$ 
and $|B|=m$, a local minimum under translation of the $L_1$-bidirectional or 
$L_\infty$-bidirectional RMS distance between $A$ and $B$, can be computed in 
time $O((n\log{m}+m\log{n})\log{\min{\{n,m\}}})$.
\end{theorem}

\subsection{\bf Minimum Hausdorff RMS distance under translation in two dimensions}
Again, we focus on the unidirectional case, but the analysis and results
extend in a straightforward manner to the bidirectional variants. Recall that,
for two sets $A=\{a_1,\ldots,a_n\}$ and $B=\{b_1,\ldots,b_m\}$ in $\R^2$, the 
\emph{minimum unidirectional Hausdorff RMS distance under translation} from 
$B$ to $A$ is
\[
\h_T(B,A) = \min_{t \in \R^2}\h(B+t,A) =\min_{t\in \R^2} \sum_{i=1}^m \left\| b_i + t - N_{A,i}(t)\right\|^{2}.
\]

As before, we put $r(t)=\h(B+t,A)$, for $t \in{\R^2}$, and we seek a translation 
$t^*\in\R^2$ that brings $r$ to a local minimum. 
Here, one can also compute a local minimum by applying the two-dimensional version 
of the ICP algorithm, but in the worst case it might perform $O(n^2m^2)$ 
iterations, each taking $O(m\log{n})$ time \cite{ICP2006}. 
Moreover, one can calculate the \emph{global} minimum of $r(t)$ in 
$O(n^2m^2\log{(nm)})$ time, as follows.

Let $\VD(A)$ denote the Voronoi diagram of $A$, and let $M$ denote the overlay 
subdivision of the $m$ shifted Voronoi diagrams, $\VD(A-b_i)$, for $b_i\in B$, which are just copies of 
$\VD(A)$, shifted by the corresponding $b_i\in{B}$. $M$ has $O(n^2m^2)$ 
regions, and this bound is tight in the worst case \cite{ICP2006}. $M$ can be constructed 
in $O(n^2m^2\log{(nm)})$ time, using, e.g., a standard line-sweep technique.

Within each region $\tau$ of $M$, the nearest-neighbor assignments $N_{A,i}(t)$, 
for $i=1,\ldots,m$, are fixed for all $t \in{\tau}$. Hence, the graph of $r(t)$ 
over $\tau$ is a portion of a single paraboloid of the form 
$r(t)=m\| t \|^2+ \scalprod{\dd_\tau }{t} + c_\tau$, for a suitable vector $\dd_\tau$ and scalar
$c_\tau$. This allows us (when $c_\tau$ and $\dd_\tau$ are available) to find a local 
minimum of $r(t)$ within the interior of $\tau$ (if one exists) in constant time. 
One also needs to test for a local minimum over each edge and vertex of $M$, and 
this too can be done in constant time for each such feature, provided that the 
explicit expression for $r(t)$ over that feature is known. This expression can be 
updated in constant time as we move from one feature of $M$ to a neighboring feature. 
All this leads to the promised computation of the global minimum of $r(t)$ in 
$O(n^2m^2\log{(nm)})$ time, and our goal is to find a local minimum
much faster.

\paragraph{The high-level algorithm for finding a local minimum.}
We now present an improved algorithm that runs in time $O(nm \log^2{(nm)})$.
Similar to the one-dimensional case (and to the case of partial matching), we search for a local minimum of $r(t)$ within
a vertical slab~$I$, that we keep shrinking until it contains no vertex of $M$
in its interior. This final slab crosses $M$ in a sequence of only $O(nm)$ regions,
stacked above one another and separated by (portions of) edges of $M$. It is then
routine to scan these regions, construct the explicit expression for $r(t)$ over each region,
updating these expressions in constant time as we go from one region to an adjacent one,
and searching for the global minimum within $I$, all in $O(nm)$ time.

A main component in our approach is a procedure that decides whether $r(t)$ has a 
local minimum to the left or to the right of a vertical line $\lambda$. 
In analogy to the procedure in Section~\ref{sec:pmalg}, 
we call this procedure $\Gamma_1(\lambda)$. This decision 
can be made in $O(nm\log(nm))$ time, as follows. We first calculate the 
\emph{global} minimum of $r(t)$ restricted to $\lambda$, by intersecting $\lambda$
with the $O(nm)$ edges of the shifted Voronoi diagrams $\VD(A-b_i)$, by sorting these
intersections along $\lambda$, and by constructing the explicit expressions for $r(t)$
over each interval between two consecutive intersections, updating these
expressions in $O(1)$ time as we cross from one interval to an adjacent one. 
Having found the global minimum $\bar{t}$ along 
$\lambda$, we then inspect the sign of $\frac{\bd r}{\bd x}$ at $\bar{t}$, and go in 
the direction where $r$ is locally smaller than $r(\bar{t})$. 
(If the derivative is $0$, we have found a local minimum at $\bar{t}$ and we terminate the entire algorithm.)
All this takes $O(nm\log(nm))$ time. 
As mentioned in Section~\ref{sec:pmalg}, we should use a \emph{global} minimum of $r(t)$ restricted to $\lambda$, for guaranteeing that a local minimum indeed lies on the side of $\lambda$ that we continue with.  

We use the decision procedure $\Gamma_1$ for shrinking the slab $I$, while ensuring that 
it continues to contain a local minimum. The shrinking is done in two stages. 
The first stage narrows $I$ until it has no \emph{original} vertices of the shifted 
Voronoi diagrams inside it. The second stage narrows the slab further to a slab 
that has no vertices of $M$ (i.e., intersection points of Voronoi edges of different 
diagrams) in it. 

\paragraph{Pruning the original Voronoi vertices.}
The overlay $M$ contains $s=O(nm)$ original Voronoi vertices. 
We sort these vertices into a list $L=(v_1,\ldots,v_{s})$, by their $x$-coordinates. 

Let $\lambda_i$ be the $y$-parallel line through $v_i$, for $i=1,\ldots,s$. 
The slab that we start with is $I=[\lambda_1,\lambda_s]$.  
We run $\Gamma_1(\lambda_1)$
and $\Gamma_1(\lambda_s)$, computing the global minima on $\lambda_1$ and on $\lambda_s$, 
and inspecting the signs of $\frac{\partial r}{\partial x}$ at these minima. If the 
output points to a local minimum outside $I$,
then one of the semi-$x$-unbounded side slabs to the left or to the right of $I$
contains a local minimum and has no original Voronoi vertex in its interior, and we stop
the first stage with that slab. Otherwise, we perform a binary search over $L$,
where at each step of the search, at some vertex $v_i$, we run the decision procedure 
$\Gamma_1(\lambda_i)$ and determine which of the two sub-slabs that are split from the current 
slab by $\lambda_i$ contains a local minimum, using the rule stated above.
This stage takes $O(nm\log^2(nm))$ time.

\paragraph{Pruning the remaining vertices.}
Let $I$ denote the final slab of the previous stage. Since $I$ does not contain any 
original Voronoi vertices, every edge of any Voronoi diagram that meets $I$ crosses 
it from side to side, so its intersection with $I$ coincides with the intersection of 
the line supporting the edge. Let $S$ denote the set of these lines.

The number of intersections between the lines of $S$ within $I$ can still be large
(but at most $O(n^2m^2)$). We run a binary search over these intersections, 
to shrink $I$ further to a slab between two consecutive intersections (that contains 
a local minimum). To guide the binary search, we use the classical \emph{slope-selection}
procedure \cite{csss} that can compute, for a given slab $I$ and a given parameter $k$, 
the $k$-th leftmost intersection point of the lines in $S$ within $I$, in $O(nm\log(nm))$ time.
With this procedure at hand, the binary search performs $O(\log{(nm)})$ steps, 
each taking $O(nm\log{(nm)})$ time, both for finding the relevant intersection 
point, and for running $\Gamma_1$ at the corresponding vertical line. 
Thus, this stage takes (also) $O(nm\log^2{(nm)})$ time.

\paragraph{Computing a local minimum in the final slab.}
Once there are no vertices of $M$ within $I$, $I$ is crossed by at most $O(nm)$ 
edges of $M$, each crossing $I$ from its left line to its right line. Consequently, these 
edges partition $I$ into $O(nm)$ trapezoidal or triangular slices, each being 
a portion of a single region of $M$, and is bounded by the left and right bounding lines 
of $I$, and by two consecutive edges of $M$ (in their $y$-order).

We compute $r(t)$ in, say, the top slice, and its minimum in that slice. Then we traverse 
the slices from top to bottom, and update $r(t)$ for every slice that we encounter in 
constant time. 
In each slice, we check whether $r(t)$ has a local minimum inside the slice. 
Since, by construction, the slab must contain a local minimum, we will find it.

The running time of this final computation is comprised of computing $r(t)$ once, in 
$O(m\log{n})$ time, and afterwards updating it, in constant time, 
$O(nm)$ times, and computing its minimum in $I$. Therefore, this stage 
takes a total of $O(nm)$ time.

Thus, with all these components, we get the main result of this section:

\begin{theorem}\label{th2dim1dir}
Given two finite point sets $A$, $B$ in $\R^2$, with $|A|=n$ and $|B|=m$, a 
local minimum of the unidirectional Hausdorff RMS distance from $B$ to $A$ under 
translation can be computed in time $O(nm\log^2(nm))$.
\end{theorem}

The bidirectional variants can be handled in much the same way, and, omitting the details, we get:

\begin{theorem}\label{th2dim2dir}
Given two finite point sets $A,B$ in $\R^2$, with $|A|=n$ and ${|B|=m}$,
a local minimum of the $L_1$-bidirectional or the $L_\infty$-bidirectional Hausdorff 
RMS distance between $A$ and $B$ under translation can be computed in time $O(nm\log^2(nm))$.
\end{theorem}

\noindent\textbf{Acknowledgments.}
Rinat Ben-Avraham was supported by Grant 2012/229 from the U.S.-Israel Binational Science Foundation. Matthias Henze was supported by ESF EUROCORES programme Euro\-GIGA-VORONOI, (DFG): RO 2338/5-1. Rafel Jaume was supported by ``Obra Social La Caixa'' and the DAAD. Bal\'azs Keszegh was supported by Hungarian National Science Fund (OTKA), under grant PD 108406, NN 102029 (EUROGIGA project GraDR 10-EuroGIGA-OP-003), NK 78439, by the J\'anos Bolyai Research Scholarship of the Hungarian Academy of Sciences and by the DAAD. Orit E.~Raz was supported by
Grant 892/13 from the Israel Science Foundation. Micha Sharir was supported by Grant 2012/229 from the U.S.-Israel Binational Science Foundation, by Grant 892/13 from the Israel Science Foundation, by the Israeli Centers for Research Excellence (I-CORE) program (center no.~4/11), and by the Hermann Minkowski--MINERVA Center for Geometry at Tel Aviv University. Igor Tubis was supported by the Deutsch Institute. 
A preliminary version of this paper has appeared in~\cite{ESA}.


\end{document}